\documentclass[runningheads]{llncs}
\usepackage{graphicx}
\usepackage{enumitem}
\begin{document}
\title{Online Multi-Facility Location}
\author{Christine Markarian\inst{1} \and Abdul-Nasser Kassar\inst{2} \and Manal Yunis\inst{2}}
%\inst{2,3}\orcidID{1111-2222-3333-4444} \and
%Third Author\inst{3}\orcidID{2222--3333-4444-5555}}
\institute{Department of Mathematical Sciences, Haigazian University, Beirut, Lebanon
\and Department of Information Technology and Operations Management, Lebanese American University, Beirut, Lebanon \\
\email{christine.markarian@haigazian.edu.lb; abdulnasser.kassar@lau.edu.lb; myunis@lau.edu.lb}}
\maketitle              % typeset the header of the contribution
\begin{abstract} Facility Location problems ask to place facilities in a way that optimizes a given objective function so as to provide a service to all clients. These are one of the most well-studied optimization problems spanning many research areas such as operations research, computer science, and management science. Traditionally, these problems are solved with the assumption that clients need to be served by one facility each. In many real-world scenarios, it is very likely that clients need a robust service that requires more than one facility for each client. In this paper, we capture this robustness by exploring a generalization of Facility Location problems, called Multi-Facility Location problems, in the online setting. An additional parameter $k$, which represents the number of facilities required to serve a client, is given. We propose the first online algorithms for the metric and non-metric variants of Multi-Facility Location and measure their performance with competitive analysis, the standard to measure online algorithms, in the worst case, in which the cost of the online algorithm is compared to that of the optimal offline algorithm that knows the entire input sequence in advance.

%are to be provided service. A signiﬁcant shortcoming of this formulation is that a few very distant clients, called outliers, can exert a disproportionately strong inﬂuence over the ﬁnal solution. In this paper we explore a generalization of various facility location problems (K-center, K-median, uncapacitated facility location etc) to the case when only a speciﬁed fraction of the customers are to be served. What makes the problems harder is that we have to also select the subset that should get service.

\end{abstract}

\section{Introduction}

\emph{Facility Location} (FL) is a classical NP-hard optimization problem widely studied in the fields of computer science and operations research~\cite{FLbookapplications,manne1964plant}. In its simplest form, we are given a set of facilities and a set of clients. Each facility has an opening cost and each client $i$ has a connecting cost to each facility $j$, which is the distance between $i$ and $j$. The goal is to open a subset of the facilities and connect the clients to open facilities so as to minimize the sum of the facility costs and the connecting costs. FL is known as two versions, \emph{metric} and \emph{non-metric}. In the metric version, the distances are assumed to be symmetric and satisfy the triangle inequality. 

We consider the \emph{online} setting in which clients are not known in advance but revealed to the algorithm over time. As soon as one arrives, it needs to be connected. Many real-world applications that contain FL as a sub-problem have this online nature in which one is expected to react to present demands whenever they arrive, without knowing about future demands. Maintaining a given optimization goal in the face of this uncertainty becomes more challenging and encourages the study of designing \emph{online algorithms}~\cite{fiat1998online} rather than classical offline algorithms for FL problems. The standard framework to measure online algorithms is \emph{competitive analysis}, in which demands and their arrival order are selected by an
oblivious adversary, that is unaware of the choices of the algorithm. An online algorithm is $c$-competitive or has competitive ratio $c$ if for all sequences of demands, the cost incurred by the algorithm is at most $c$ times the cost incurred by an optimal offline algorithm, which knows the entire sequence of demands in advance.

The study of \emph{Facility Location} (FL) in the online setting initiated with Meyerson~\cite{meyerson2001online}, who introduced the metric \emph{Online Facility Location} problem (OFL), and proposed an $\mathcal{O}(\log n)$-competitive randomized algorithm, where $n$ is the number of clients. Alon~$\emph{et al.}$~\cite{Alon:2006:GAO:1198513.1198522} studied the non-metric version and proposed an $\mathcal{O}(\log n \log m)$-competitive randomized algorithm, where $n$ is the number of clients and $m$ is the number of facilities. Many other variations were known for both metric and non-metric variants in the online setting~\cite{abshoff2016towards,anagnostopoulos2004simple,diveki2011online,fotakis2007primal,fotakis2008competitive,markarianonline}.

All of these works assume that clients need to be served by one facility each. In many real-world applications, a robust service, in which a client is served by more than one facility, is desirable. Facilities and/or connections to facilities may be prone to failure and assigning clients to multiple facilities would provide a fault-tolerant solution, such as providing replicated cash data in distributed networks~\cite{Byrka2010}.

\subsection{Our Contribution} 

In this paper, we capture this robustness by exploring a generalization of online \emph{Facility Location}~\cite{Alon:2006:GAO:1198513.1198522,fotakis2008competitive}, in which we are additionally given a parameter $k$, which is the number of facilities required to serve a client. We refer to this generalization as \emph{online multi-facility location} and define it as follows. 

\begin{definition} (Online multi-facility location) We are given a collection of $m$ facilities, $n$ clients, and a positive integer $k$. Each facility has an \emph{opening cost} and each client has a \emph{connecting cost} to each facility. Clients arrive over time. As soon as one arrives, it needs to be connected to at least $k$ open facilities. To open a facility, we pay its opening cost. To connect a client to a facility, we pay the corresponding connecting cost.  The goal is to minimize the total opening and connecting costs. 
\end{definition}

We address the metric and non-metric online multi-facility location problems. As far as we are aware of, there are no online algorithms in the literature that solve these variants or can be trivially extended to solve them. 

\paragraph{Lower bounds.} For the non-metric version, there is a lower bound of $\Omega(\log m\log n)$, under the assumption that NP $\not\subseteq$ BPP, where $m$ is the number of facilities and $n$ is the number of clients, due to the lower bound given for \emph{Online Set Cover}~\cite{Korman}. As for the metric version, there is a lower bound of $\Omega(\frac{\log n}{\log \log n})$, where $n$ is the number of clients, due to the lower bound given for metric OFL~\cite{fotakis2008competitive}.

\begin{enumerate}
\item We refer to the non-metric version as \emph{Online Non-metric Multi-Facility Location} (ONMFL). We propose an online $\mathcal{O}(\log (kn) \log m)$-competitive randomized algorithm for ONMFL, where $m$ is the number of facilities; $n$ is the number of clients; and $k$ is the number of required connections. The latter uses a \emph{randomized rounding} approach that first constructs a fractional solution and then rounds it into an integral one. Its competitive ratio is analyzed with simple arguments based on first comparing the fractional solution constructed by the algorithm to the optimal offline solution and then measuring the fractional solution in terms of the integral one. 

\item We refer to the metric version as \emph{Online Metric Multi-Facility Location} (OMMFL). We propose an online $\mathcal{O}(\max \{ \frac{f_{max}}{f_{min}}, \frac{c_{max}}{c_{min}} \} \cdot k  \cdot \frac{\log n}{\log \log n})$-competitive deterministic algorithm for OMMFL, where $n$ is the number of clients; $k$ is the number of required connections; $c_{max}$ and $c_{min}$ are the maximum and minimum connecting costs, respectively; $f_{max}$ and $f_{min}$ are the maximum and minimum facility costs, respectively.
%We study two variants of OMMFL, \emph{uniform} OMMFL, in which all facilities have unit cost, and \emph{non-uniform} OMMFL, in which facilities have different costs.  
The idea of the algorithm is to ensure first that each client is connected to one facility by running an algorithm for metric \emph{Online Facility Location} (OFL). Then, the $k -1$ remaining connections are made by choosing the cheapest possible facilities so as not to worsen the competitive ratio by much. Our approach can be seen as a general framework that transforms any given online algorithm for metric OFL into an algorithm for OMMFL by losing a bounded factor in the competitive ratio.

%\begin{itemize}

%\item We propose an online $\mathcal{O}(\frac{c_{max}}{c_{min}} \cdot k \cdot \frac{\log n}{\log \log n})$-competitive deterministic algorithm for uniform OMMFL, where $n$ is the number of clients; $k$ is the number of required connections; $c_{max}$ and $c_{min}$ are the maximum and minimum connecting costs, respectively.

%\item We propose an online $\mathcal{O}(\max \{ \frac{f_{max}}{f_{min}}, \frac{c_{max}}{c_{min}} \} \cdot k  \cdot \frac{\log n}{\log \log n})$-competitive deterministic algorithm for non-uniform OMMFL, where $n$ is the number of clients; $k$ is the number of required connections; $c_{max}$ and $c_{min}$ are the maximum and minimum connecting costs, respectively; $f_{max}$ and $f_{min}$ are the maximum and minimum facility costs, respectively.
%\end{itemize}
\end{enumerate}

%The performance of \emph{online algorithms} is measured by their \emph{competitive ratio}, which is the worst case ratio of the total cost of the online algorithm to that of an optimal offline algorithm that knows the entire input sequence in advance. The performance of \emph{randomized online algorithms} is measured by comparing the total \emph{expected cost} of the algorithm, against an adversary that is not aware of the random choices of the algorithm, knows the entire input sequence in advance, and is optimal. 

%\section{Related Work}
%Berman and DasGubta in \cite{BERMAN200854} gave an online randomized algorithm for OSMC, based on the randomized winnowing approach of \cite{10.1023/A:1022869011914}. They also gave deterministic lower bounds for the unweighted and weighted variants of OSMC, based on the approaches of~\cite{Alon:2003:OSC:780542.780558}. 

\section{Online Non-metric Multi-Facility Location} In this section, we present an online randomized algorithm for ONMFL and analyze its competitive ratio.

\subsection{Preliminaries \& Related Work}

ONMFL is a  generalization of the \emph{Online Non-metric Facility Location} (ONFL)~\cite{Alon:2006:GAO:1198513.1198522} with $k =1$. Alon~$\emph{et al.}$~\cite{Alon:2006:GAO:1198513.1198522} gave an $\mathcal{O}(\log n \log m)$-competitive randomized algorithm for ONFL, where $n$ is the number of clients and $m$ is the number of facilities. 

A closely related problem is the \emph{Online Set k-Multicover} (OSMC)~\cite{BERMAN200854}. Given a universe $\mathcal {U}$ of $n$ elements, a family $\mathcal {S}$ of $m$ subsets of $V$, each associated with a cost, and a positive integer $k$. A subset $D \subseteq \mathcal {U}$ of elements arrives over time. OSMC asks to select a collection $\mathcal {C} \subseteq \mathcal {S}$ of subsets, of minimum cost, such that each arriving element belongs to at least $k$ subsets of $\mathcal {C}$. Berman and DasGubta~\cite{BERMAN200854} proposed an $\mathcal{O}(\log m \log d)$-competitive randomized algorithm for OSMC, where $m$ is the number of subsets and $d$ is the maximum set size.

Transformations between OSMC and ONMFL instances can be made in both directions. An instance of ONMFL can be transformed into an instance of OSMC as follows. We associate each facility with each of the $2^n -1$ possible groups of clients, and let each facility/group be a subset, with cost equal to the sum of the cost of the facility and the connecting costs of the clients in the group to the facility. We let each client be an element. Following this transformation, the algorithm of Berman and DasGubta would imply a feasible algorithm for ONMFL, but with competitive ratio $\mathcal{O}(\log (m(2^n -1)) \log n)$, where $n$ is the number of clients, and $m$ is the number of facilities. An instance of OSMC can be transformed into an instance of ONMFL as follows. We represent each subset by a facility and let the opening cost be the subset cost. We represent each element by a client and set the connecting cost to 0 if the client belongs to the subset/facility and infinity otherwise.The offline version of OSMC, in which all elements are given in advance, has an $\mathcal{O}(\log d)$-approximation~\cite{JOHNSON1974256,vazirani2001}, where $d$ is the maximum set size. Similar transformations can be made in the offline setting. This implies an $\mathcal{O}(\log n)$-approximation for the offline version of ONMFL, where $n$ is the number of clients. These are the best achievable unless $NP \subseteq DTIME(n^{\log \log n})$~\cite{10.1145/285055.285059}. The \emph{Online Set Cover} (OSC)~\cite{Alon:2003:OSC:780542.780558} is a special case of OSMC with $k = 1$. Korman's lower bound on the competitive ratio of OSC~\cite{Korman} implies an $\Omega(\log m\log n)$ lower bound on the competitive ratio of ONMFL, under the assumption that NP $\not\subseteq$ BPP, where $m$ is the number of facilities and $n$ is the number of clients.

\subsection{Online Algorithm}

Our algorithm for ONMFL is based on constructing a fractional solution first and then rounding it online into an integral solution. Unlike \cite{Alon:2003:OSC:780542.780558} and \cite{Alon:2006:GAO:1198513.1198522} for \emph{Online Set Cover} (OSC) and \emph{Online Non-metric Facility Location} (ONFL), respectively, in which a potential function is used in the competitive analysis of the algorithms, our analysis is based on simple arguments, similar to those we gave in our previous work for a related problem, the \emph{Online Node-weighted Steiner Forest}~\cite{DBLP:conf/iwoca/Markarian18}. 

We start by giving the following graph formulation for ONMFL. Given a root node $r$, $m$ facility nodes, and $n$ client nodes. We add an edge from $r$ to each facility node $j$, with cost equal to the opening cost of facility $j$. We add an edge from each facility node $j$ to each client node $i$, with cost equal to the connecting cost of client $i$ to facility $j$. As soon as a client $i$ arrives, we need to purchase $k$ disjoint paths between $r$ and $i$. The goal is to minimize the total costs of the paths purchased, where the cost of a path is the cost of its edges. To output a solution for ONMFL, each facility whose corresponding edge is purchased will be opened, and each client whose corresponding edge to an open facility is purchased will be connected to that facility.  

The algorithm initially knows $n$, the number of clients; $k$, the number of required connections; and the opening costs of facilities. A subset $D$ of $n' \leq n$ clients arrives over time. As soon as a client $i$ arrives, the algorithm is given the connecting costs of $i$ to each facility, and is expected to react.

We consider the graph formulation described earlier. Let $r$ be a root node; each facility is represented as a facility node, and has an edge to $r$, associated with its opening cost. Upon the arrival of a new client $i$, the algorithm creates a client node for it and adds an edge from this node to each facility node, associated with the given connecting cost. Let $G = (V, E)$ be this graph. 

Each edge added to $E$ is given a fraction, set intitially to 0. The algorithm does not allow these fractions to decrease over time. These form a \emph{fractional solution} for ONMFL. The \emph{maximum flow} between node $u$ and node $v$ in $G$ is the smallest total weight of edges which if removed would disconnect $u$ from $v$. These edges form a \emph{minimum cut} between $u$ and $v$ in $G$. To compute the maximum flow/minimum cut between two nodes in a graph, we run the algorithm of Shchroeder~$\emph{et al.}$ in \cite{Schroeder}. Let $c_e$  and $f_e$ be the cost and fraction of edge $e$, respectively. A path is purchased if and only if its edges are purchased.
\paragraph{Random process.} The algorithm makes its random choices, based on $\alpha$, the minimum among $2\left\lceil \log(kn+1) \right\rceil$ independently chosen random variables, distributed uniformly in the interval $\left[0, 1\right]$. 

\noindent Next, we describe how the algorithm reacts upon the arrival of a new client. 
\newline
\newline
{\bf Input:} $G = (V, E)$ and client node $i \in D$
\newline
{\bf Output:} Set of edges purchased
\newline
\newline
Make a copy $G'$ of $G$; \newline
As long as there are $< k$ disjoint paths purchased between $r$ and $i$ in $G$, do the following:
\begin{enumerate}
\item \indent While the maximum flow between $r$ and $i$ in $G'$ is less than 1, construct a minimum cut $Q$ between $r$ and $i$ in $G'$; for each edge $e \in Q$, make the following fraction increase:
$$f_{e} = f_{e} \cdot (1 + 1/c_{e}) + \frac{1}{\left| Q \right| \cdot c_{e}}$$
\item Purchase each edge $e$ with $f_e > \alpha$.
\item If there is no purchased path between $r$ and $i$ in $G'$, find a minimum-cost such path and purchase it. 
\item Refer to all facilities whose corresponding edges were purchased as \emph{open}; delete from $G'$ the purchased edges between $i$ and each open facility. 
\end{enumerate}
%\hrulefill

\subsection{Competitive Analysis}

The algorithm buys edges in the second and third steps. In the second step, its choices are made based on the random process, whereas in the third step, its choices are made to guarantee a feasible solution. We measure the expected cost of each separately.  Let $Opt$ be the cost of the optimal offline solution and let $frac$ be the cost of the \emph{fractional} solution constructed by the algorithm in the first step. 

\paragraph{Choices based on random process:} Let $S'$ be the set of edges purchased in the second step of the algorithm and let $C_{S'}$ be its expected cost. These edges are purchased by the algorithm based on the random process described earlier. Let us fix an $l: 1 \leq l \leq 2\left\lceil \log(kn+1) \right\rceil$ and an edge $e$. We denote by $X_{e, l}$ the indicator variable of the event that $e$ is chosen by the algorithm based on the random choice of $l$.   
\begin{equation}
\label{frac3}
C_{S'} = \sum_{e \in S'} \sum_{l = 1}^{2\left\lceil \log(kn+1) \right\rceil} c_e \cdot Exp \left[X_{e, l} \right]  = 2\left\lceil \log(kn+1) \right\rceil  \sum_{e \in S'} c_e  f_{e} 
\end{equation} 

Notice that $\sum_{e \in S'} c_e f_{e}$ is upper bounded by the cost of the fractional solution. The latter can be measured against the optimal offline solution, as follows. The idea here is that every time the algorithm performs a fraction increase, it does not exceed 2. Moreover, the total number of fraction increases can be measured in terms of the cost of the optimal offline solution. 

The fraction increase contributed by each edge $e$ in a minimum cut $Q$ is $\left( \frac{f_e}{c_e} + \frac{1}{\left| Q  \right| \cdot c_e} \right)$. The algorithm would make a fraction increase only if the maximum flow is less than 1. This means we have that $\sum_{e \in Q} f_{e} < 1$ before a fraction increase. Therefore, each fraction increase does not exceed: 

\begin{equation}
\label{frac2}
\sum_{e \in Q} c_e \cdot \left( \frac{f_{e}}{c_e} + \frac{1}{\left| Q  \right| \cdot c_e} \right) < 2
\end{equation}

As long as the algorithm hasn't purchased at least $k$ disjoint paths between $r$ and a given client $i$, it enters the loop that starts by constructing a maximum flow on the graph $G'$. Notice how $G'$ shrinks over time, as the algorithm purchases the paths. 

\begin{lemma}
\label{aboutG'}
Whenever the algorithm makes a fraction increase, $G'$ contains at least one path that is also in the optimal offline solution. 
\end{lemma}
\begin{proof}
To see why this holds, fix any time $t$ before a fraction increase. Let $s < k$ be the number of disjoint paths purchased by the algorithm at time $t$. $G'$ at time $t$ must contain at least one optimal path since $G'$ is constructed by removing $s$ (less than $k$) feasible paths from $G$ and the optimal offline solution contains at least $k$ disjoint paths in $G$. \qed
\end{proof}

Finally, the algorithm would have an edge $e$ from the optimal offline solution in every minimum cut $Q$ of $G'$, since $Q$ must contain an edge from each path, by definition.  Based on the equation for the fraction increase, after $\mathcal{O}(\log |Q|) $ fraction increases, the fraction $f_e$ of $e$ becomes 1, and no further increases can be made, as $e$ will not be in any future minimum cut. The size of any minimum cut is upper bounded by $m$, the number of facilities or the maximum available paths between $r$ and client $i$. We can now bound the fractional solution: 

\begin{equation}
\label{frac1}
frac \leq \mathcal{O}(\log m \cdot Opt) 
\end{equation}

Combining Equations \ref{frac3}, \ref{frac2}, and \ref{frac1} imply an upper bound on the expected cost $C_{S'}$ of the edges bought in the second step of the algorithm:

\begin{equation}
\label{frac}
C_{S'} \leq \mathcal{O}(\log (kn) \log m \cdot Opt) 
\end{equation}

\paragraph{Choices to guarantee feasible solution:} Let $S''$ be the set of edges purchased in the third step of the algorithm and let $C_{S''}$ be its expected cost. These edges are purchased by the algorithm only if a path has not been bought by the random process in the second step. Every time the algorithm purchases a path in this step, its cost does not exceed $Opt$ since the algorithm buys the minimum-cost path in $G'$, and as we showed earlier in Lemma \ref{aboutG'}, $G'$ contains at least one path that is also in the optimal solution.

\begin{itemize}
\item (one client, one path) We start by calculating the expected cost incurred by a single client for purchasing a single path. Fix a client $i$. Let $Q_{j+1}$ be a minimum cut of $G'$ constructed after the algorithm has purchased $j < k$ disjoint paths between $r$ and $i$ and has completed the first step. The probability of purchasing the $(j + 1)^{th}$ path for a single $1 \leq l \leq 2\left\lceil \log(kn+1) \right\rceil$ is: 

$$ \prod_{e \in Q_{j+1}}  (1 - f_{e}) \leq e^{-\sum_{e  \in Q_{j+1}} f_{e}} \leq 1/e $$

Notice that the last inequality holds because the algorithm ensures that $\sum_{e  \in Q_{j+1}} f_{e}\geq 1$ at the end of the first step (Max-flow min-cut theorem). The expected cost of purchasing the $(j + 1)^{th}$ path for all $1 \leq l \leq 2\left\lceil \log(kn+1) \right\rceil$, is less than $1/(kn)^2 \cdot Opt$. 

\item (one client, $k$ paths) The expected cost of purchasing all $k$ paths is the sum of the expected costs for each path and is less than $k \cdot 1/(kn)^2 \cdot Opt$.  

\item (total cost of all clients) The total expected cost incurred by all $n'$ clients that arrived is less than:

$$n' \cdot k \cdot 1/(kn)^2 \cdot Opt \leq n \cdot k \cdot 1/(kn)^2 \cdot Opt =1/kn \cdot Opt $$. 

Therefore, the expected cost $C_{S''}$ of the edges bought in the third step of the algorithm is:

\begin{equation}
\label{integ}
C_{S''} \leq 1/kn \cdot Opt 
\end{equation}

\end{itemize}

Equations \ref{frac} and \ref{integ} yield to the following theorem.

\begin{theorem} There is an online $\mathcal{O}(\log (kn) \log m)$-competitive randomized algorithm for the Online Non-metric Multi-Facility Location, where $m$ is the number of facilities, $n$ is the number of clients, and $k$ is the number of required connections. 
\end{theorem}

\section{Online Metric Multi-Facility Location} 

In this section, we present an online deterministic algorithm for OMMFL and analyze its competitive ratio. While this problem has been intensively studied in the offline setting~\cite{Byrka2010,YAN2011545}, we are not aware of any online algorithm for it. 
%Recall that in uniform OMMFL, all facilities have unit cost, whereas in non-uniform OMMFL, facilities have different costs. In the following, we present online deterministic algorithms for uniform and non-uniform OMMFL and analyze their competitive ratio. 

%in which clients and facilities correspond to points in a metric space, where distances between them are non-negative, symmetric, and satisfy the triangle inequality.

\subsection{Preliminaries \& Related Work}

OMMFL is a generalization of metric \emph{Online Facility Location} (OFL)~\cite{fotakis2007primal,fotakis2008competitive,meyerson2001online} with $k = 1$. Meyerson~\cite{meyerson2001online} introduced metric OFL and proposed an $\mathcal{O}(\log n)$-competitive randomized algorithm, where $n$ is the number of clients. Fotakis~\cite{fotakis2008competitive} showed that no randomized online algorithm can achieve a competitive ratio better that $\Omega(\frac{\log n}{\log \log n})$ against an oblivious adversary and gave a deterministic algorithm with asymptotically matching $\mathcal{O}(\frac{\log n}{\log \log n})$-competitive ratio. In another work later, he proposed a primal-dual deterministic algorithm with $\mathcal{O}(\log n)$-competitive ratio, that was simpler to formulate, analyze, and implement~\cite{fotakis2007primal}. The competitive ratio we achieve for OMMFL is based on running the deterministic algorithm of Fotakis~\cite{fotakis2008competitive} for metric OFL.

The lower bound on the competitive ratio of metric OFL by Fotakis~\cite{fotakis2008competitive} implies an $\Omega(\frac{\log n}{\log \log n})$ lower bound on the competitive ratio of OMMFL.

\subsection{Online Algorithm}

Let $A_{OFL}$ be any online (deterministic or randomized) algorithm for metric \emph{Online Facility Location} (OFL), with competitive ratio $r$. Given an instance $I$ of OMMFL with positive integer $k$. Client $i$ arrives. Our algorithm needs to connect $i$ to $k$ different open facilities. 

\begin{enumerate}

\item The algorithm starts by running $A_{OFL}$ on instance $I$, where $k = 1$. This results in opening some facilities and connecting $i$ to one open facility. 
\item If $i$ is the first client, we open the cheapest $k-1$ facilities other than the one $i$ is connected to. From this point on, there are at least $k$ open facilities. Until $A_{OFL}$ opens these facilities itself, these remain closed with respect to $A_{OFL}$. 
\item The algorithm will then connect $i$ to any other $k-1$ open facilities. 
\end{enumerate}

\subsection{Competitive Analysis}

Let $I$ be an instance of OMMFL with positive integer $k$. Let $I'$ be the same instance as $I$ except for $k = 1$. Let $Opt$ and $Opt'$ be the cost of the optimal solution for $I$ and that for $I'$, respectively. Let $C$ and $C'$ be the cost of our algorithm for $I$ and that of $A_{OFL}$ for $I'$, respectively. We denote by $C_{fac}$ and $C_{con}$ the costs incurred by our algorithm to open facilities and to connect clients, respectively. We denote by $C'_{fac}$ and $C'_{con}$ the costs incurred by $A_{OFL}$ to open facilities and to connect clients, respectively.

The algorithm opens the cheapest $k-1$ facilities other than the ones opened by $A_{OFL}$. Let $f_{max}$ be the maximum facility cost and $f_{min}$ the minimum facility cost. Thus, we have that: 

%The total cost of the cheapest $k-1$ faculties bought by the algorithm can be additionally upper bounded by $(k-1) f_{max}$. Hence, 
%Since the optimal solution must contain at least $k$ open facilities, we have that the cheapest $k-1$ facilities are upper bounded by $Opt$. 

\begin{equation}
C_{fac} \leq C'_{fac} + f_{max} \cdot (k - 1)
\end{equation}

Given any client $i$. Apart from its connecting cost $c_i$ incurred by $A_{OFL}$, our algorithm connects $i$ to $k - 1$ other facilities, each resulting in a connecting cost at most $\frac{c_{max}}{c_{min}} \cdot c_i $, where $c_{max}$ and $c_{min}$ are the maximum and minimum connecting costs, respectively. This implies an overall connecting cost: 

\begin{equation}
C_{con} \leq C'_{con} \cdot (1 + \frac{c_{max}}{c_{min}}(k - 1))
\end{equation}

We now add the two equations above and do some algebraic manipulations by using: $C'_{con} \leq C'$, $C' _{fac} \leq C'$, $C' \geq c_{min}$, and $C' \geq f_{min}$ to get: 

 \begin{equation}
C \leq C' \cdot (2 + \frac{f_{max}}{f_{min}} (k-1) + \frac{c_{max}}{c_{min}} (k-1))
\end{equation}

Recall that $A_{OFL}$ is $r$-competitive and so $C' \leq r \cdot Opt'$. Since $Opt' \leq Opt$, the theorem below follows.

\begin{theorem} Given an online (deterministic or randomized) $r$-competitive algorithm for metric Online Facility Location. Then there is an online \\ $\mathcal{O}(\max \{ \frac{f_{max}}{f_{min}}, \frac{c_{max}}{c_{min}} \} \cdot k  \cdot r)$-competitive algorithm for the Online Metric Multi-Facility Location, where $k$ is the number of required connections; $c_{max}$ and $c_{min}$ are the maximum and minimum connecting costs, respectively; $f_{max}$ and $f_{min}$ are the maximum and minimum facility costs, respectively.
\end{theorem}

Running the deterministic algorithm of Fotakis~\cite{fotakis2008competitive} for metric OFL, with $\mathcal{O}(\frac{\log n}{\log \log n})$-competitive ratio, results in the following. 

\begin{corollary}
There is an online $\mathcal{O}(\max \{ \frac{f_{max}}{f_{min}}, \frac{c_{max}}{c_{min}} \} \cdot k  \cdot \frac{\log n}{\log \log n})$-competitive deterministic algorithm for the Online Metric Multi-Facility Location, where $n$ is the number of clients; $k$ is the number of required connections; $c_{max}$ and $c_{min}$ are the maximum and minimum connecting costs, respectively; $f_{max}$ and $f_{min}$ are the maximum and minimum facility costs, respectively.
\end{corollary}

%from 2) Facility Location provides a simple and natural model for network design and data clustering problems and has been the subject of intensive research over the last decade (see e.g. [12] for a survey and [6] for approximation algorithms and applications).In addition to the ofﬂine setting, there are many practical applications where either the demand points are not known in advance or the solution must be constructed incrementally using limited (if any) information about future demands (see e.g. [11] for some examples from the areas of network design and data clustering).

\section{Concluding Remarks \& Future Work} 
%The table below gives a summary of the results obtained for ONMFL, uniform OMMFL, and non-uniform OMMFL. 

%begin{table}
%\caption{Competitive ratios of ONMFL, uniform OMMFL, non-uniform OMMFL}\label{tab1}
%\begin{tabular}{l|l|l}
%\hline
%Problem &  Lower bound & Upper bound\\
%\hline
%ONMFL &  (randomized) $\Omega(\log m\log n)$ & (randomized) $\mathcal{O}(\log (kn) \log m)$\\
%Uniform OMMFL &  $\Omega(\frac{\log n}{\log \log n})$  & $\mathcal{O}(\frac{c_{max}}{c_{min}} \cdot k \cdot \frac{\log n}{\log \log n})$\\
%Non-uniform OMMFL & $\Omega(\frac{\log n}{\log \log n})$ & $\mathcal{O}(\max \{ \frac{f_{max}}{f_{min}}, \frac{c_{max}}{c_{min}} \} \cdot k  \cdot \frac{\log n}{\log \log n})$\\
%\hline
%\end{tabular}
%\end{table}

In this paper, we have assumed there is a unique positive integer $k$ for all clients. In many application scenarios, it is likely that each client has different number of required connections. A slight modification in our algorithms would yield to $\mathcal{O}(\log (k_{max}n) \log m)$ and $\mathcal{O}(\max \{ \frac{f_{max}}{f_{min}}, \frac{c_{max}}{c_{min}} \} \cdot k_{max}  \cdot \frac{\log n}{\log \log n})$ competitive ratios for ONMFL and OMMFL, respectively, where $k_{max}$ is the maximum required connections. One research direction is to target competitive ratios independent of $k_{max}$, $k_{min}$, or even $k$.

%A slight modification in our randomized algorithm for \emph{Online Non-metric Multi-Facility Location} (ONMFL) yields to an $\mathcal{O}(\log (k_{max}n) \log m)$-competitive randomized algorithm, where $k_{max}$ is the maximum required connections. As for \emph{Online Metric Multi-Facility Location} (OMMFL), the algorithm and its analysis can be extended to the case with different number of required connections, where $k_{max}$ will replace $k$ and appear in the competitive ratio instead of $k$. 

This brings us to the next question, for \emph{Online Metric Multi-Facility Location} (OMMFL) about whether we can get rid of the parameters $c_{max}$, $c_{min}$, $f_{max}$, and $f_{min}$ from the competitive ratio or achieve lower bounds in terms of these parameters. To achieve the former, one may want to attempt a primal-dual approach for instance, by trying to extend the algorithm of Fotakis~\cite{fotakis2007primal} for metric \emph{Online Facility Location}.

Finally, demands and their arrival order in this paper are given by an oblivious adversary. Assuming these are given according to some probability distribution, it might be possible to design online algorithms with better competitive ratios.

%The following modifications are made to achieve this result.

%\begin{itemize}

%\item The algorithm need not know in advance all the different values of $k$, but rather $k_{max}$.  

%\item The random process is slightly modified as follows. The algorithm makes its random choices, based on $\alpha$, the minimum among $2\left\lceil \log(k_{max}n+1) \right\rceil$ independently chosen random variables, distributed uniformly in the interval $\left[0, 1\right]$. 

%\item Given a client $i$. Let $k_i$ be the number of minimum required connections for $i$. The algorithm performs the four steps as long as there are $< k_i$ disjoint paths purchased between $r$ and $i$ in $G$. 

%\item In the competitive analysis, a slight modification is required when calculating the expected cost resulting from the choices made by the algorithm to guarantee a feasible solution. The expected cost of purchasing all $k_i$ paths is the sum of the expected costs for each path and is less than $k_i \cdot 1/(k_{max}n)^2 \cdot Opt \leq k_{max} \cdot 1/(k_{max}n)^2 \cdot Opt$.  
%\end{itemize}

\bibliographystyle{splncs04}

\bibliography{bibliography}

\begin{thebibliography}{10}
\providecommand{\url}[1]{\texttt{#1}}
\providecommand{\urlprefix}{URL }
\providecommand{\doi}[1]{https://doi.org/#1}

\bibitem{abshoff2016towards}
Abshoff, S., Kling, P., Markarian, C., Meyer auf~der Heide, F., Pietrzyk, P.:
  Towards the price of leasing online. Journal of Combinatorial Optimization
  \textbf{32}(4),  1197--1216 (2016)

\bibitem{Alon:2003:OSC:780542.780558}
Alon, N., Awerbuch, B., Azar, Y.: The online set cover problem. In: Proceedings
  of the Thirty-fifth Annual ACM Symposium on Theory of Computing. pp.
  100--105. STOC '03, ACM, New York, NY, USA (2003).
  \doi{10.1145/780542.780558}, \url{http://doi.acm.org/10.1145/780542.780558}

\bibitem{Alon:2006:GAO:1198513.1198522}
Alon, N., Awerbuch, B., Azar, Y., Buchbinder, N., Naor, J.S.: A general
  approach to online network optimization problems. ACM Trans. Algorithms
  \textbf{2}(4),  640--660 (Oct 2006). \doi{10.1145/1198513.1198522},
  \url{http://doi.acm.org/10.1145/1198513.1198522}

\bibitem{anagnostopoulos2004simple}
Anagnostopoulos, A., Bent, R., Upfal, E., Van~Hentenryck, P.: A simple and
  deterministic competitive algorithm for online facility location. Information
  and Computation  \textbf{194}(2),  175--202 (2004)

\bibitem{BERMAN200854}
Berman, P., DasGupta, B.: Approximating the online set multicover problems via
  randomized winnowing. Theoretical Computer Science  \textbf{393}(1),  54 --
  71 (2008). \doi{https://doi.org/10.1016/j.tcs.2007.10.047},
  \url{http://www.sciencedirect.com/science/article/pii/S030439750700847X}

\bibitem{Byrka2010}
Byrka, J., Srinivasan, A., Swamy, C.: Fault-tolerant facility location: A
  randomized dependent lp-rounding algorithm. In: Eisenbrand, F., Shepherd,
  F.B. (eds.) Integer Programming and Combinatorial Optimization. pp. 244--257.
  Springer Berlin Heidelberg, Berlin, Heidelberg (2010)

\bibitem{diveki2011online}
Div{\'e}ki, G., Imreh, C.: Online facility location with facility movements.
  Central European Journal of Operations Research  \textbf{19}(2),  191--200
  (2011)

\bibitem{FLbookapplications}
Drezner, Z.: Facility Location - A Survey of Applications and Methods. Springer
  Series in Operations Research and Financial Engineering (1995)

\bibitem{10.1145/285055.285059}
Feige, U.: A threshold of ln n for approximating set cover. J. ACM
  \textbf{45}(4),  634?652 (Jul 1998). \doi{10.1145/285055.285059},
  \url{https://doi.org/10.1145/285055.285059}

\bibitem{fiat1998online}
Fiat, A., Woeginger, G.J.: Online algorithms: The state of the art, vol.~1442.
  Springer (1998)

\bibitem{fotakis2007primal}
Fotakis, D.: A primal-dual algorithm for online non-uniform facility location.
  Journal of Discrete Algorithms  \textbf{5}(1),  141 -- 148 (2007)

\bibitem{fotakis2008competitive}
Fotakis, D.: On the competitive ratio for online facility location.
  Algorithmica  \textbf{50}(1),  1 -- 57 (2008)

\bibitem{JOHNSON1974256}
Johnson, D.S.: Approximation algorithms for combinatorial problems. Journal of
  Computer and System Sciences  \textbf{9}(3),  256 -- 278 (1974).
  \doi{https://doi.org/10.1016/S0022-0000(74)80044-9},
  \url{http://www.sciencedirect.com/science/article/pii/S0022000074800449}

\bibitem{Korman}
Korman, S.: {On the use of randomization in the online set cover problem}.
  Master's thesis, Weizmann Institute of Science, Israel (2005)

\bibitem{manne1964plant}
Manne, A.S.: Plant location under economies-of-scale—decentralization and
  computation. Management Science  \textbf{11}(2),  213--235 (1964)

\bibitem{DBLP:conf/iwoca/Markarian18}
Markarian, C.: An optimal algorithm for online prize-collecting node-weighted
  steiner forest. In: Combinatorial Algorithms - 29th International Workshop,
  IWOCA 2018, Singapore, July 16-19, 2018, Proceedings. pp. 214 -- 223 (2018).
  \doi{10.1007/978-3-319-94667-2\_18},
  \url{https://doi.org/10.1007/978-3-319-94667-2\_18}

\bibitem{markarianonline}
Markarian, C., Meyer auf~der Heide, F.: Online algorithms for leasing vertex
  cover and leasing non-metric facility location. In: Parlier, G.H.,
  Liberatore, F., Demange, M. (eds.) Proceedings of the 8th International
  Conference on Operations Research and Enterprise Systems, {ICORES} 2019,
  Prague, Czech Republic, February 19-21, 2019. pp. 315--321. SciTePress
  (2019). \doi{10.5220/0007369503150321},
  \url{https://doi.org/10.5220/0007369503150321}

\bibitem{meyerson2001online}
Meyerson, A.: Online facility location. In: Proceedings 42nd IEEE Symposium on
  Foundations of Computer Science. pp. 426 -- 431. IEEE (2001)

\bibitem{Schroeder}
Schroeder, J., Guedes, A., P.~Duarte~Jr, E.: Computing the minimum cut and
  maximum flow of undirected graphs. Tech. rep., Federal University of
  Paran\'{a}, Department of Informatics (2004)

\bibitem{vazirani2001}
Virkumar, V.V.: Approximation Algorithms. Springer-Verlag (2001)

\bibitem{YAN2011545}
Yan, L., Chrobak, M.: Approximation algorithms for the fault-tolerant facility
  placement problem. Information Processing Letters  \textbf{111}(11),  545 --
  549 (2011). \doi{https://doi.org/10.1016/j.ipl.2011.03.005},
  \url{http://www.sciencedirect.com/science/article/pii/S0020019011000688}

\end{thebibliography}

\end{document}